\documentclass[11pt]{article}

\usepackage{graphicx}
\usepackage {amssymb}
\usepackage {amsmath}
\usepackage{subfigure}
\usepackage{color}
\usepackage{url}

\newtheorem {theorem} {Theorem}

\newtheorem {lemma} {Lemma}

\newtheorem {corollary} {Corollary}

\newenvironment {proof}{\textbf {Proof:}}{\hfill \ensuremath {\boxtimes} \vspace{11pt}}

\graphicspath{{figs/}}

\newcommand{\Reals}{\mathbb{R}}
\newcommand{\Sphere}{\mathbb{S}}
\newcommand{\setL}{\ensuremath{L}}

\title{Balanced partitions of $3$-colored geometric sets in the plane\thanks{This paper was published in Discrete Applied Mathematics, 181:21--32, 2015~\cite{journal}.}}

\date{}
\author{Sergey Bereg\thanks{University of Texas at Dallas, USA. besp@utdallas.edu.} \and 
Ferran Hurtado\thanks{Dept. de  Matem\'atica Aplicada II, Universitat Polit\`{e}cnica de Catalunya, Spain. $\{$ferran.hurtado, carlos.seara, rodrigo.silveira$\}$@upc.edu.} \and 
Mikio Kano\thanks{Ibaraki University, Japan. kano@mx.ibaraki.ac.jp.} \and 
Matias Korman\thanks{National Institute of Informatics, Japan.  korman@nii.ac.jp. JST, ERATO, Kawarabayashi Large Graph Project.} \and 
Dolores Lara\thanks{CINVESTAV-IPN, Dept. de Computaci{\'o}n, Mexico. dlara@cs.cinvestav.mx.} \and 
Carlos Seara\footnotemark[3] \and 
Rodrigo I. Silveira\footnotemark[3] $^{,}$\thanks{Dept. de  Matem\'atica \& CIDMA, Universidade de Aveiro, Portugal. rodrigo.silveira@ua.pt.} \and 
Jorge Urrutia\thanks{Instituto de Matem\'aticas, UNAM, Mexico. urrutia@matem.unam.mx.} \and 
Kevin Verbeek\thanks{University of California, Santa Barbara, USA. kverbeek@cs.ucsb.edu.}}

\begin{document}

\maketitle

\begin{abstract}
Let $S$ be a finite set of geometric objects partitioned into classes or \emph{colors}.  A subset
$S'\subseteq S$ is said to be \emph{balanced} if $S'$ contains the same amount of elements of $S$ from each of the colors.
We study several problems on partitioning $3$-colored sets of points and lines in the plane into two balanced subsets:
(a) We prove that for every 3-colored arrangement of lines there exists a segment that intersects exactly one line of each color, and that when there are $2m$ lines of each color, there is a segment intercepting $m$ lines of each color.
(b) Given $n$ red points, $n$ blue points and $n$ green points on any closed Jordan curve $\gamma$, we show that for every integer $k$ with $0 \leq k \leq n$ there is a pair of disjoint intervals on $\gamma$ whose union contains exactly $k$ points of each color.
(c) Given a set $S$ of $n$ red points, $n$ blue points and $n$ green points in the integer lattice satisfying certain constraints,
there exist two rays with common apex,  one vertical and one horizontal, whose union splits the plane
into two regions, each one containing a balanced subset of $S$.
\end{abstract}

\section{Introduction}\label{section:introduction}

Let $S$ be a finite set of geometric objects distributed into classes or \emph{colors}.  A subset
$S_1\subseteq S$ is said to be \emph{balanced} if $S_1$ contains the same amount of elements of $S$ from each of the colors.
Naturally, if $S$ is balanced, its complement is also balanced, hence we talk of a {\em balanced bipartition} of $S$.

When the point set $S$ is in the plane, and the balanced partition is defined by a geometric object $\zeta$ splitting the plane into two regions, we say that $\zeta$ is \emph{balanced} (and {\em nontrivial} if both regions contain points of $S$).
A famous example of such a partition is the discrete version of the \emph{ham-sandwich theorem}: given a set of $2n$ red points and $2m$ blue
points in general position in the plane, there always exists a line $\ell$ such that each halfplane bounded by $\ell$ contains exactly $n$ red points and $m$ blue points.
It is well known that this theorem can be generalized to higher dimensions and can  be formulated in terms of splitting continuous measures.

There are also plenty
of variations of the ham-sandwich theorem. For example, it has been proved that given $gn$ red points and $gm$ blue points in the plane in general position, there exists a subdivision of the plane into $g$ disjoint convex polygons, each of which contains $n$ red points and $m$ blue points \cite{Sergei2000}. Also, it was shown in  \cite{Imre2001} (among other results) that for any two measures in the plane there are $4$ rays with common apex such that each of the sectors they define contains $\frac{1}{4}$ of both measures.
For many more extensions and detailed results we refer the interested reader to~\cite{jorge1,jorge2}, the survey \cite{Kaneko2003} of Kaneko and Kano and to the book \cite{Matousek} by Matou{\v s}ek.

Notice that if we have a $3$-colored set of points $S$ in the plane, it is possible that no line produces any non-trivial balanced partition of $S$. Consider for example an equilateral triangle $p_1p_2p_3$ and replace every vertex $p_i$ by a very small disk $D_i$ (so that no line can intersect the three disks), and place
$n$ red points, $n$ green points, and $n$ blue points, inside the disks $D_1$, $D_2$ and $D_3$, respectively. It is clear for this configuration that no line determines a halfplane containing exactly $k$ points of each color, for any value of $k$ with $0<k<n$.

However, it is easy to show that for every $3$-colored set of points $S$ in the plane there is a conic that simultaneously bisects the three colors: take the plane to be $z=0$ in $\mathbb{R}^3$, lift the points
vertically to the unit paraboloid $P$, use the 3-dimensional ham-sandwich theorem for splitting evenly the lifted point set with a plane $\Pi$, and use the projection of $P\cap\Pi$ as halving conic in $z=0$. On the other hand, instead of changing the partitioning object, one may impose some additional constraints on the point set. For example, Bereg and Kano have recently proved that if all vertices of the convex hull of $S$ have the same color, then there exists a nontrivial balanced line~\cite{Bereg2012}. This result was recently extended to sets of points in a space of higher dimension by Akopyan and Karasev \cite{Karasev2012}, where the constraint imposed on the set was also generalized.

\medskip

\noindent\textbf{Our contribution}. In this work we study several problems on balanced bipartitions of $3$-colored sets of points and lines in the plane.
In Section \ref{section:lines} we prove that for every 3-colored arrangement of lines, possibly unbalanced, there always exists a segment intersecting exactly one line of each color. If the number of lines of each color is exactly $2n$, we show that there is always a segment intersecting exactly $n$ lines of each color. The existence of balanced segments in 3-colored line arrangements is equivalent,  by duality, to the existence of balanced double wedges in $3$-colored point sets.  

In Section \ref{jordan} we consider balanced partitions on closed Jordan curves. Given $n$ red points, $n$ blue points and $n$ green points on any closed Jordan curve $\gamma$, we show that for every integer $k$ with $0 \leq k \leq n$ there is a pair of disjoint intervals on $\gamma$ whose union contains exactly $k$ points of each color.

In Section \ref{section:lattice} we focus on point sets in the integer plane lattice $\mathbb{Z}^2$; for simplicity, we will refer to $\mathbb{Z}^2$ as \emph{the lattice}.
We define an \emph{$L$-line} \emph{with }{\em corner} $q$ as the union of two different rays with common apex $q$, each of them being either vertical or horizontal. This {\em $L$-line} partitions the plane into two regions (Figure~\ref{fig:L-line}). If one of the rays is vertical and the other ray is horizontal, the regions are a quadrant with origin at $q$ and its complement. Note, however, that we allow an $L$-line to consist of two horizontal or two vertical rays with opposite direction, in which case the $L$-line is simply a horizontal or vertical line that splits the plane into two halfplanes. An $L$-line segment can be analogously defined using line segments instead of rays.

$L$-lines in the lattice play somehow a role comparable to the role of ordinary lines in the real plane. An example of this is the result
due to Uno et al.~\cite{Kano2009}, which extends the ham-sandwich theorem to the following scenario:
Given $n$ red points and $m$ blue points in general position in $\mathbb{Z}^2$, there always exists an $L$-line that bisects both sets of points. This result was also generalized by Bereg \cite{Bereg2009}; specifically he proved that for any integer $k\ge 2$ and for any $kn$ red points and $km$ blue points in general position in the plane, there
exists a subdivision of the plane into $k$ regions using at most $k$ horizontal segments and at most $k-1$ vertical segments such that every region contains $n$ red points and $m$ blue points. Several results on sets of points in $\mathbb{Z}^2$,  using $L$-lines or $L$-line segments are described in \cite{Kano2012}.

A set $S\subset \mathbb{R}^2$ is said to be {\em orthoconvex} if the intersection of $S$ with every horizontal or vertical line is connected.
The {\em orthogonal convex hull} of a set $S$ is the intersection of all connected orthogonally convex supersets of S.

Our main result in Section \ref{section:lattice} is in correspondence with the result of Bereg and Kano \cite{Bereg2012} mentioned above that if the convex hull of a $3$-colored point
set is monochromatic, then it admits some balanced line. Specifically, we prove here that
given
a set $S \subset \mathbb{Z}^2$ of $n$ red points, $n$ blue points and $n$ green points in general position (i.e., no two points are horizontally or vertically aligned), whose orthogonal convex hull is monochromatic,
then there is always an $L$-line that separates a region of the plane
containing exactly $k$ red points, $k$ blue points, and $k$ green points from $S$, for
some integer $k$ in the range $1\le k \le n-1$.

We conclude in Section \ref{section:conclusion} with some open problems and final remarks.

\section{3-colored line arrangements}\label{section:lines}

Let $\setL = R \cup G \cup B$ be a set of lines in the plane, such that $R, G$ and $B$ are pairwise disjoint. We refer to the elements of $R$, $G$, and $B$ as red, green, and blue, respectively. Let $\mathcal{A}(\setL)$ be the arrangement induced by the set $\setL$. We assume that $\mathcal{A}(\setL)$ is \emph{simple}, i.e., there are no parallel lines and no more than two lines intersect at one point. In Section~\ref{subsection:cells} we first prove that there always exists a face in $\mathcal{A}(\setL)$ that contains all three colors. We also  extend this result to higher dimensions. We say that a segment is \emph{balanced} with respect to $\setL$ if it intersects the same number of red, green and blue lines of $\setL$. In Section~\ref{subsection:doubleWedges} we prove that (i) there always exists a segment intersecting exactly one line of each color; and (ii) if the size of each set $R, G$ and $B$ is $2n$, there always exists a balanced segment intersecting $n$ lines of each color. As there are standard duality transformations between points and lines in which segments correspond to double wedges, the results in this section can be rephrased in terms of the existence of balanced double wedges for $3$-colored point sets.

\subsection{Cells in colored arrangements}\label{subsection:cells}

In this section we prove that there always exists a $3$-colored face in $\mathcal{A}(\setL)$, that is, a face that has at least one side of each color. In fact, we can show that a $d$-dimensional arrangement of $(d-1)$-dimensional hyperplanes, where each hyperplane is colored by one of $d+1$ colors (at least one of each color), must contain a $(d+1)$-colored cell. This result is tight with respect to the number of colors. If we have only $d$ colors, then every cell containing the intersection point of $d$ hyperplanes with different colors is $d$-colored. On the other hand, it is not difficult to construct examples of arrangements of hyperplanes with $d+2$ colors where no $(d+2)$-colored cell exists.

\begin{figure}[tb]
\centering
\includegraphics[scale=1.2]{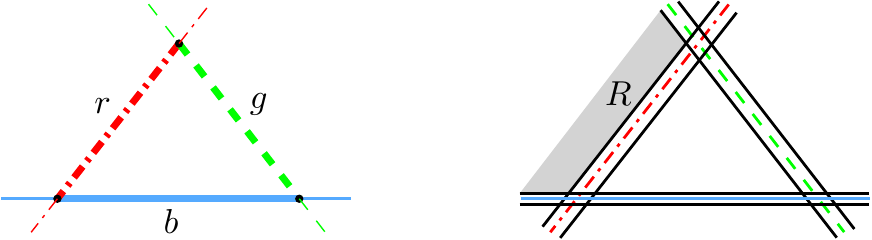}
\caption{
Construction of a line arrangement with no $4$-colored cell.}
\label{fig:counterexample}
\end{figure}
An example in the plane is shown in Figure~\ref{fig:counterexample}. Start with a triangle in which each side has a different color (red, green, or blue), and extend the sides to the colored lines $r$, $g$ and $b$ that support the sides (Figure~\ref{fig:counterexample}, left). Then shield each of these lines by two black parallel lines, one on each side (Figure~\ref{fig:counterexample}, right). Finally perturb the black lines in such a way that the arrangement becomes simple, yet the intersection points between former parallel lines are very far away. Now it is easy to see that no cell can contain all four colors. For example, depending on the specific intersections of the black lines with $g$ and $b$, the region $R$ may contain the colors green and blue, but cannot contain color red.

This example can be generalized to $d$-dimensional space. Start with a $d$-simplex in which each of the $d+1$ hyperplanes supporting a facet has a different color, $c_1,\dots,c_{d+1}$. Then shield each facet with two parallel hyperplanes having color $c_{d+2}$, one on each side, and perturb as above to obtain a simple arrangement in which no cell is $(d+2)$-colored.

For intuition's sake, before generalizing the result to higher dimensions, we first prove the result for $d=2$. Consider the $2$-dimensional arrangement $\mathcal{A}(\setL)$ as a graph. The dual of a face $f$ of $\mathcal{A}(\setL)$ is a face $\hat{f}$ that contains a vertex for every bounding line of $f$, and contains an edge between two vertices of $\hat{f}$ if and only if the intersection of the corresponding lines is part of the boundary of $f$ (see Figure~\ref{fig:ColorArrangement}(a)).

\begin{figure}
\centering
\includegraphics[scale=1.2]{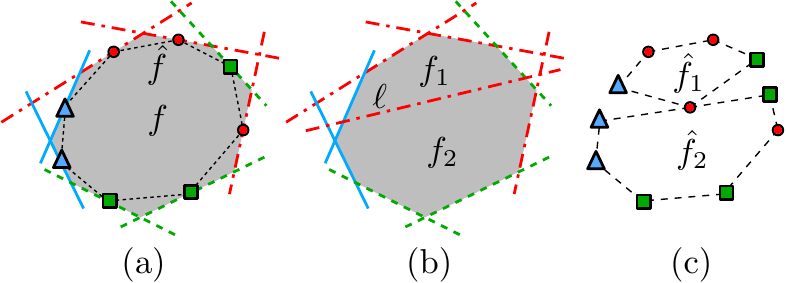}
\caption{
(a) A (complete) face $f$ and its dual $\hat{f}$ (dashed).
(b) Face $f$ is split into $f_1$ and $f_2$.
(c) The dual $\hat{f}$ is split into $\hat{f}_1$ and $\hat{f}_2$.}
\label{fig:ColorArrangement}
\end{figure}

Let  $C$  be a simple cycle of vertices where each vertex is colored either red, green, or blue. Let $n_r(C), n_g(C), $ and $n_b(C)$ be the number of red, green, and blue vertices of $C$, respectively. We simply write $n_r$, $n_g$, and $n_b$ if $C$ is clear from the context. The \emph{type} of an edge of $C$ is the multiset of the colors of its vertices. Let $n_{rr}$, $n_{gg}$, $n_{bb}$, $n_{rg}$, $n_{rb}$, and $n_{gb}$ be the number of edges of the corresponding type. Note that, if $f$ is bounded, then $\hat{f}$ is a simple cycle, where each vertex is colored either red, green, or blue. We say a bounded face $f$ is \emph{complete} if $n_{rg} \equiv n_{rb} \equiv n_{gb} \equiv 1 \pmod{2}$ holds for $\hat{f}$.

\begin{lemma}
\label{lem:3colorcycle}
Consider a simple cycle in which each vertex is colored either red, green, or blue. Then $n_{rg} \equiv n_{rb} \equiv n_{gb} \pmod{2}$.
\end{lemma}
\begin{proof}
The result follows from double counting. We can obtain an expression for (twice) the number of vertices of a certain color by summing up over all edges that have vertices of that color. 
For instance, for $n_r$ we get the equation $2 n_r = 2 n_{rr} + n_{rg} + n_{rb}$. This directly implies that $n_{rg} \equiv n_{rb} \pmod{2}$. By repeating the same process for the other colors we obtain the claimed result.
\end{proof}

\begin{theorem}
\label{thm:3ColoredFace}
Let $\setL$ be a set of 3-colored lines in the plane inducing a simple arrangement $\mathcal{A}(\setL)$, such that each color appears at least once. Then there exists a complete face in $\mathcal{A}(\setL)$.
\end{theorem}
\begin{proof}
The result clearly holds if $|R|=|G|=|B|=1$.
For the general case, we start with one line of each color, and then incrementally add the remaining lines, maintaining a complete face $f$ at all times.
Without loss of generality, assume that a red line $\ell$ is inserted into $\mathcal{A}(\setL)$. If $\ell$ does not cross $f$, we keep $f$. Otherwise, $f$ is split into two faces $f_1$ and $f_2$ (see Figure~\ref{fig:ColorArrangement}(b)). Similarly, $\hat{f}$ is split into $\hat{f}_1$ and $\hat{f}_2$ (with the addition of one red vertex, see Figure~\ref{fig:ColorArrangement}(c)). Because $\ell$ is red, the number $n_{gb}$ of green-blue edges does not change, that is, $n_{gb}(\hat{f}) = n_{gb}(\hat{f}_1) + n_{gb}(\hat{f}_2)$. This implies that either $n_{gb}(\hat{f}_1)$ or $n_{gb}(\hat{f}_2)$ is odd. By Lemma~\ref{lem:3colorcycle} it follows that either $f_1$ or $f_2$ is complete.
\end{proof}

We now extend the result to higher dimensions. For convenience we assume that every hyperplane is colored with a ``color'' in $[d] = \{0, 1, \ldots, d\}$. Consider a triangulation $T$ of the surface of the $(d-1)$-dimensional sphere $\Sphere^{d-1}$, where every vertex is colored with a color in $[d]$. Note that a triangulation of $\Sphere^1$ is exactly a simple cycle. As before, we define the \emph{type} of a simplex (or face) of $T$ as the multiset $S$ of the colors of its vertices. Furthermore, let $n_S$ be the number of simplices (faces) with type $S$. We say a type $S$ is \emph{good} if $S$ does not contain duplicates and $|S| = d$. The following analogue of Lemma~\ref{lem:3colorcycle} is similar to Sperner's lemma~\cite{Sperner96}.

\begin{lemma}
\label{lem:dColorTriang}
Consider a triangulation $T$ of $\Sphere^{d-1}$, where each vertex is colored with a color in $[d]$. Then either $n_S \equiv 0 \pmod{2}$ for all good types $S$, or $n_S \equiv 1 \pmod{2}$ for all good types $S$.
\end{lemma}
\begin{proof}
We again use double counting, following the ideas in the proof of Lemma~\ref{lem:3colorcycle}. Consider a subset $S \subset [d]$ with $|S| = d-1$ (thus $S$ includes all but two colors from $[d]$). 
Recall that good types consist of $d$ different colors, so there are exactly two good types $S_1$ and $S_2$ that contain $S$: one with each of the two colors of $[d]$ that are not already present in $S$. 

Since every $(d-2)$-dimensional face of $T$ is present in exactly two $(d-1)$-dimensional simplices, we can obtain an expression for (twice) $n_S$ by summing over all face types that contain $S$:
\[
2 n_S = n_{S_1} + n_{S_2} + 2 \sum_{x \in S} n_{S \uplus \{x\}},
\]

where the summation on the right is done over the types that contain $S$ but are not good (the symbol $\uplus$ denotes the \emph{disjoint union}, to allow duplicated colors in the face type).

From the above equation we obtain that $n_{S_1} \equiv n_{S_2} \pmod{2}$. By repeating this procedure for every set $S$, we obtain the claimed result.
\end{proof}

Consider a cell $f$ of a $d$-dimensional arrangement of $(d-1)$-dimensional hyperplanes. The dual $\hat{f}$ of $f$ contains a vertex for every bounding hyperplane of $f$, and contains a simplex on a set of vertices of $\hat{f}$ if and only if the intersection of the corresponding hyperplanes is part of the boundary of $f$. Note that, if $f$ is bounded, then $\hat{f}$ is a triangulation of $\Sphere^{d-1}$. We say a bounded cell $f$ is \emph{complete} if $n_S(\hat{f}) \equiv 1 \pmod{2}$ for all good types $S$. Note that a complete cell is $(d+1)$-colored.

\begin{theorem}
\label{thm:dColoredCell}
Let $\setL$ be a set of $(d+1)$-colored hyperplanes in $\Reals^d$ inducing a simple arrangement $\mathcal{A}(\setL)$, such that each color appears at least once. Then there exists a complete face in $\mathcal{A}(\setL)$.
\end{theorem}
\begin{proof}
The proof is analogous to the two-dimensional case: if there is exactly one hyperplane of each color, then the arrangement has exactly one bounded face $f$, which must be complete (this can be easily shown by induction on $d$). We maintain a complete face $f$ during successive insertions of hyperplanes. Assume we add a hyperplane $H$ with color $x$. If $H$ does not cross $f$, we can simply keep $f$. Otherwise, $f$ is split into two faces $f_1$ and $f_2$, and $\hat{f}$ is split into $\hat{f}_1$ and $\hat{f}_2$ (with the addition of one vertex of color $x$). Let $S = [d] - \{x\}$. Because $H$ has color $x$, the number of simplices with type $S$ does not change, that is, $n_S(\hat{f}) = n_S(\hat{f}_1) + n_S(\hat{f}_2)$. This implies that either $n_S(\hat{f}_1)$ or $n_S(\hat{f}_2)$ is odd. By Lemma~\ref{lem:dColorTriang} it follows that either $f_1$ or $f_2$ is complete.
\end{proof}

An immediate consequence of Theorem~\ref{thm:dColoredCell} is the following result:

\begin{corollary}
\label{cor:segment111}
Let $\setL$ be a $(d+1)$-colored set of hyperplanes in $\Reals^d$ inducing a simple arrangement $\mathcal{A}(\setL)$, such that each color appears at least once.
Then there exists a segment intersecting exactly one hyperplane of each color.
\end{corollary}
\begin{proof}
Consider a $(d+1)$-colored cell. By Theorem~\ref{thm:dColoredCell} such a cell must exist and it must also contain an intersection of $d$ hyperplanes with different colors. Now we can take the segment from this intersection to a face of the remaining color (in the same cell). By perturbing and slightly extending this segment, we obtain a segment properly intersecting exactly one hyperplane of each color.
\end{proof}

\subsection{$3$-colored point sets and balanced double wedges}\label{subsection:doubleWedges}

We now return to the plane and consider $3$-colored point sets. By using the point-plane duality, Corollary~\ref{cor:segment111} implies the following result.
\begin{theorem}\label{thm:doubleWedge111}
Let $S$ be a 3-colored set of points in $\Reals^2$ in general position, such that each color appears at least once.
Then there exists a double wedge that contains exactly one point of each color from $S$.
\end{theorem}

\begin{proof}
We apply the standard duality transformation between points and non-vertical lines where a point $p=(a,b)$ is mapped to a line $p'$ with equation $y=ax-b$, and vice versa~\cite{deberg}. By Corollary~\ref{cor:segment111}, there exists a segment $w$ that intersects exactly one line of each color. By standard point-line duality properties, the dual of $w$ is a double wedge $w'$ that contains the dual points of the intersected lines.
\end{proof}

Since the dual result extends to higher dimensions, so does the primal one. The equivalent statement says that given a set of points colored with $d+1$ colors in $\Reals^d$, there exists a {\em pencil} (i.e., a collection of hyperplanes sharing an affine subspace of dimension $d-2$) containing exactly one point of each color. 

Next we turn our attention to balanced $3$-colored point sets, and prove a ham-sandwich-like theorem for double wedges.

\begin{theorem}\label{th:bisectingdoubleWedge}
Let $S$ be a 3-colored balanced set of $6n$ points in $\Reals^2$ in general position.
Then there exists a double wedge that contains exactly $n$ points of each color from $S$.
\end{theorem}
\begin{proof}
We call a double wedge satisfying the theorem {\em bisecting}. Without loss of generality we assume that the points of $S$ have distinct $x$-coordinates and distinct $y$-coordinates. For two distinct points $a$ and $b$ in the plane, let $\ell(a,b)$ denote the line passing through them. Consider the arrangement $\mathcal{A}$ of all the lines passing through two points from $S$, i.e.
$$\mathcal{A}=\{ \ell(p_i,p_j)~|~p_i,p_j\in S, i\ne j\} .$$

\begin{figure}
\begin{center}
  \includegraphics{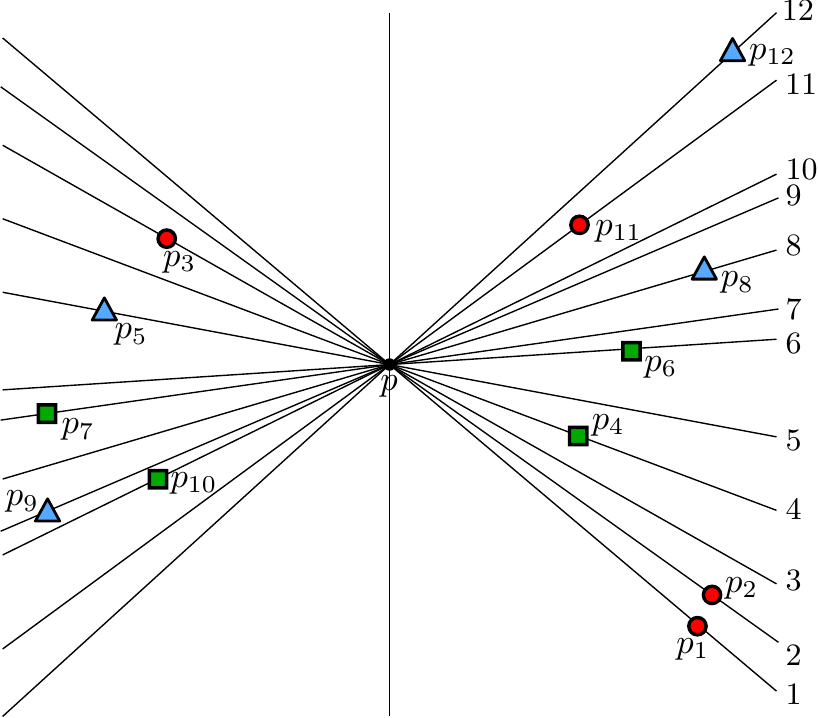}\\
  \caption{$\sigma_p$: ordering of $S$ based on the slopes of lines through $p$.} \label{ordering}
\end{center}
\end{figure}

Consider a vertical line $\ell$ that does not contain any point from $S$. We continuously walk on $\ell$ from $y=+\infty$ to $y=-\infty$. For any point $p\in \ell$ we define an ordering $\sigma_p$ of $S$ as follows: consider the lines $\ell(p,q),q\in S$ and sort them by (increasing values) of slope. Let $(p_1,\ldots, p_{6n})$ be the obtained ordering (see Figure~\ref{ordering}). 
 
By construction, any consecutive interval $\{p_i,p_{i+1},\dots,p_j\}$ of an ordering of $p$ corresponds to a set of lines whose points can be covered by a double wedge with apex at $p$ (even if the indices are taken modulo $6n$).
Likewise, for any $p\in\Reals^2$, any double wedge with apex at $p$ will appear as an interval in the ordering $\sigma_p$.

Given an ordering
$\sigma_p=(p_1, \ldots, p_{6n})$ of $S$, we construct a polygonal curve as follows: for every $k \in \{1, 2,\ldots, 6n\}$ let $b_k$ and $g_k$ be the number of blue and green points in the set
$S(p,k)=\{p_k,p_{k+1},\ldots,p_{k+3n-1}\}$ of $3n$ points, respectively.
We define the corresponding lattice point $q_k := (b_k -n, g_k-n)$, and the  polygonal curve
$\phi(\sigma)=(q_1,\ldots, q_{6n},-q_1,\ldots,-q_{6n},q_1)$.

\begin{figure}[tb]
\begin{center}
  \includegraphics{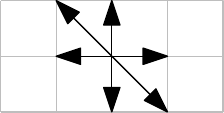}\\
  \caption{The seven types of segments ${q_{k-1}q_{k}}$ (including the segment of length 0 if $p_{k-1}$ and $p_k$ have the same color) depending on the color of $p_{k-1}$ and $p_k$.} \label{vector2}
\end{center}
\end{figure}

Intuitively speaking, the point $q_k$ indicates how balanced the interval is that starts with point $p_k$ and contains $3n$ points. By construction, if $q_k = (0,0)$ for some $p \in \ell$ and some $k\leq 6n$, then the associated wedge is balanced (and {\em vice versa}).
In the following we show that this property must hold for some $k \in \{1, \ldots, 6n\}$ and $p\in\mathbb{R}^2$.  We observe several important properties of $\phi(\sigma)$:

\begin{enumerate}
\item Path $\phi(\sigma)$ is centrally symmetric (w.r.t. the origin). This follows from the definition of $\phi$.

\item\label{noedge} Path $\phi(\sigma)$ is a closed curve. Moreover, the interior of any edge $e_i=q_iq_{i+1}$ of $\phi(\sigma)$ cannot contain the origin: consider the segment between two consecutive vertices $q_{k-1}$ and $q_k$ of $\phi$. Observe that the double wedges associated to $q_{k-1}$ and $q_k$ share $3k-1$ points. Thus, the orientation and length of the segment $q_{k-1}q_k$ only depend on the color of the two points that are not shared. In particular, there are only $7$ types of such segments in $\phi(\sigma)$, see Figure~\ref{vector2}. Since these segments do not pass through grid points, the origin cannot appear at the interior of a segment.

\item\label{nointerior} If the orderings of two points $p$ and $p'$ are equal, then their paths $\phi(\sigma_p)$ and $\phi(\sigma_{p'})$ are equal. If the orderings $\sigma_{p}$ and $\sigma_{p'}$ are not equal then either $(i)$ there is a line of $\mathcal{A}$ separating  $p$ and $p'$ or $(ii)$ there is a point $p_i\in S$ such that the vertical line passing through $p_i$ separates $p$ and $p'$.

In our proof we will move $p$ along a vertical line, so case $(ii)$ will never occur. Consider a continuous vertical movement of $p$, and consider the two orderings $\pi_1$ and $\pi_2$ before and after a line of $\mathcal{A}$ is crossed. Observe that the only difference between the two orderings is that two consecutive points (say, $p_i$ and $p_{i+1}$) of $\pi_1$ are reversed in $\pi_2$. Thus, the only difference between the two associated
$\phi$ curves will be in vertices whose associated interval contains one of the two points (and not the other). Since these two points are consecutive, this situation can only occur at vertices $q_{i+1-3n}$ and $q_{i+1}$ (recall that, for simplicity, indices are taken modulo $6n$). Since the predecessor and successor of these vertices are equal in both curves, the difference between both curves will be two quadrilaterals (and two more in the second half of the curve).

We claim that the interior of any such quadrilateral can never contain the origin. Barring symmetries, there are two possible ways in which the quadrilateral is formed, depending on the color of $p_i$ and $p_{i+1}$ (see  Figure \ref{crossing1}). Regardless of the case, the interior of any such quadrilateral cannot contain any lattice point, and in particular cannot contain the origin.

\item\label{nonzero} Path $\phi(\sigma)$ has a vertex $q_i$ such that $q_i=(0,0)$, or the curve $\phi(\sigma)$ has nonzero {\em winding} with respect to the origin. Intuitively speaking, the winding number of a closed curve $\mathcal{C}$ with respect to a point measures the net number of clockwise revolutions that a point traveling on $\mathcal{C}$ makes around the given point (see a formal definition in~\cite{winding}). By Properties~\ref{noedge} and~\ref{nointerior}, the only way in which $\phi(\sigma)$ passes through the origin is through a vertex $q_i$. If this does not happen, then no point of the curve passes through the origin. Recall that $\phi(\sigma)$  is a closed continuous curve that contains the origin and is centrally symmetric around that point. In topological terms, this is called an {\em odd} function. It is well known that these functions have odd winding (see for example~\cite{topology}, Lemma 25). In particular, we conclude that the winding of $\phi(\sigma)$ cannot be zero.
\end{enumerate}

\begin{figure}
\begin{center}
  \includegraphics{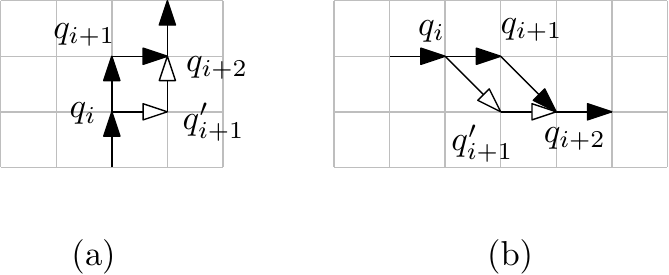}\\
  \caption{The change of $\phi(\sigma)$ when $p_i$ and $p_{i+1}$ are swapped.
(a) $p_i$ is blue and $p_{i+1}$ is green. (b) $p_i$ is red and $p_{i+1}$ is green.} \label{crossing1}
\end{center}
\end{figure}

Thus, imagine a moving point $p\in \ell$ from $y=+\infty$ to $y=-\infty$: when $p$ is located sufficiently low along $\ell$, the $y$ coordinates of the points of $S$ can be ignored, and the resulting order will give first the points to the left of $\ell$ (sorted in decreasing value of the $x$ coordinates) and then the points to the right of $\ell$ (also in decreasing value of the $x$ coordinates). Similarly, the order we obtain when $p$ is sufficiently high will be the exact reverse (see Figure~\ref{fig_extremes}).

\begin{figure}[tb]
\begin{center}
  \includegraphics{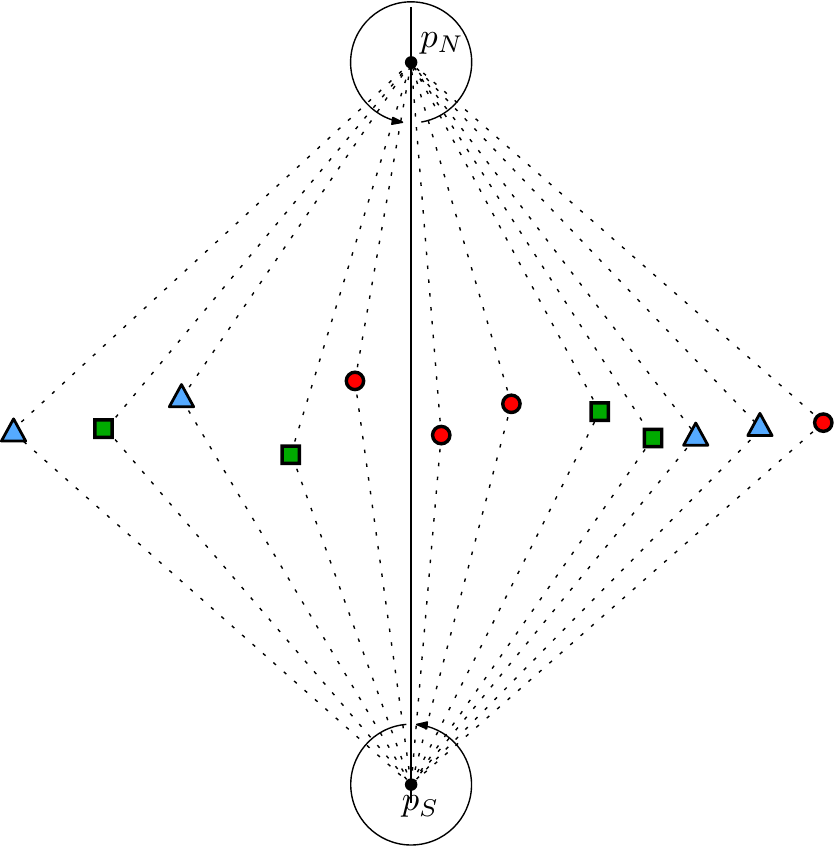}\\
  \caption{When $p$ is placed sufficiently low (or high), the sorting $p_\sigma$ gives the points ordered according to their $x$ coordinate. The orderings corresponding to points $p_N$ and $p_S$ are depicted with an arrow. In particular, notice that the two orderings are reversed.} \label{fig_extremes}
\end{center}
\end{figure}

Now consider the translation of $p$ from both extremes and the changes that may happen to the ordering along the translation. By Property~\ref{noedge}, the curve $\phi(p)$ will change when we cross a line of $\mathcal{A}$, but the differences between two consecutive curves will be very small. In particular, the space between the two curves cannot contain the origin. Consider now the instants of time in which point $p$ is at $y=+\infty$ and $y=-\infty$: if either curve contains the origin as a vertex, we are done (since such vertex is associated to a balanced double wedge). Otherwise, we observe that the orderings must be the reverse of each other, which in particular implies that the associated curves describe exactly the same path, but in reverse direction. By Property~\ref{nonzero} both curves have nonzero winding, which in particular it implies that they have opposite winding numbers (i.e. their winding number gets multiplied by -1). Since the winding must change sign, we conclude that at some point in the translation the curve $\phi(p)$ passed through the origin (Otherwise, by Hopf's degree Theorem~\cite{needham} they would have the same winding). By Properties~\ref{noedge} and~\ref{nointerior} this can only happen at a vertex of $\phi(p)$, implying the existence of a balanced double wedge.
\end{proof}

Using the point-line duality again, we obtain the equivalent result for balanced sets of lines.

\begin{corollary}
\label{cor:halvingSegment}
Let $\setL$ be a $3$-colored balanced  set of $6n$ lines in $\Reals^2$ inducing a simple arrangement. Then, there always exists a segment intersecting exactly $n$ lines of each color.
\end{corollary}

\section{Balanced partitions on closed Jordan curves} \label{jordan}
In this section we consider balanced $3$-colored point sets on closed Jordan curves. Our aim is to find a bipartition of the set that is balanced and that can be realized by at most two disjoint intervals of the curve. To prove the claim we use the following arithmetic lemma:

\begin{lemma}
\label{lem:ArithmeticLemma}
For a fixed integer $n \geq 2$, any integer $k \in \{1, 2, \dots , n \}$ can be obtained from $n$ by applying functions $f(x) = \lfloor x/2 \rfloor$ and $g(x) = n-x$ at most $2 \log n + O(1)$ times.
\end{lemma}
\begin{proof}
Consider a fixed value $k\leq n$. In the following we show that $k=h_t(h_{t-1}(\cdots h_1(n)\cdots ))$ for some $t\le 2\log n+O(1)$ where each $h_i$ is either $f$ or $g$.

For the purpose, we use the concept of {\em starting points}: we say that an integer $m$ is an $i$-starting point (with respect to $k$) if number $k$ can be obtained from $m$ by applying functions from $\{f,g\}$ at most $i$ times. Note that any number is always a $0$-starting point with respect to itself, and our claim essentially says that $n$ is a $(2 \log n + O(1))$-starting point with respect to any $k\leq n$.

Instead of explicitly computing all $i$-starting points, we compute a consecutive interval $\mathcal{V}_i\subseteq \{1,\ldots, n\}$ of starting points. For any $i\geq 0$, let $\ell_i$ and $r_i$ be the left and right endpoints of the interval $\mathcal{V}_i$, respectively. This interval is defined as follows: initially, we set $\mathcal{V}_0=\{k\}$. For larger values of $i$, we use an inductive definition: If $r_i<\lfloor n/2\rfloor$ we apply $f$ as the first operation. That is, any number that, after we apply $f$, falls within $\mathcal{V}_i$ should be in $\mathcal{V}_{i+1}$. Observe that this implies $\ell_{i+1}=2\ell_i$ and $r_{i+1}=2r_i+1$. If $\ell_i>\lfloor n/2\rfloor$ we apply $g$ as the first operation. In this case, we have $\ell_{i+1}=n-r_i$ and $r_{i+1}=n-\ell_i$. By construction, the fact that all elements of $\mathcal{V}_i$ are $i$-starting points implies that elements of $\mathcal{V}_{i+1}$ are $(i+1)$-starting points.

This sequence will finish at some index $j$ such that $\ell_j<\lfloor n/2\rfloor<r_j$ (that is, $\lfloor n/2\rfloor$ is a $j$-starting point). In particular, we have that $n$ is a $(j+1)$-starting point since $f(n)=\lfloor n/2\rfloor$. Thus, to conclude the proof it remains to show that $j\leq 2 \log n + O(1)$. Each time we use operator $f$ as the first operation, the length of the interval is doubled.
On the other hand, each time we use operator $g$, the length of the interval does not change.
Moreover, function $g$ is never applied twice in a row. Thus, after at most $2(\lceil\log n\rceil-1)$ steps, the size of interval $\mathcal{V}_i$ will be at least $2^{\lceil\log n\rceil-1}\geq \lceil n /2 \rceil$, and therefore must contain $\lfloor n/2\rfloor$.
\end{proof}

Now we can prove the main result of this section.
As we explain below, it is enough to prove the result for the case in which the Jordan curve $\gamma$ is the unit circle.
Let $S^1$ be the unit circle in $\mathbb{R}^2$.
Let $P$ be a $3$-colored balanced set of $3n$ points on $S^1$, and let $R$, $G$, and $B$ be the partition of $P$ into the three color classes.

Given a closed curve $\gamma$ with an injective continuous map $f:S^1\to\gamma$ and an integer $c>0$, we say that a set $Q\subseteq \gamma$ is a {\em c-arc set} if $Q=f(Q_S)$ where $Q_S$ is the union of at most $c$ closed arcs of $S^1$. Intuitively speaking, if $\gamma$ has no crossings, $c$ denotes the number of components of $Q$. However, $c$ can be larger than the number of components if $\gamma$ has one or more crossings.

\begin{theorem}\label{thm:curveIntervals}
Let $\gamma$ be a closed Jordan curve in the plane, and let $P$ be a 3-colored balanced set of $3n$  points on $\gamma$.
Then for every positive integer $k\leq n$ there exists a $2$-arc set $P_k \subseteq \gamma$ containing exactly $k$ points of each color.
\end{theorem}
\begin{proof}
Using $f^{-1}$ we can map the points on $\gamma$ to the unit circle, thus a solution on $S^1$ directly maps to a solution on $\gamma$. Hence, it suffices to prove the statement for the case in which $\gamma=S^1$.

Let $I$ be the set of numbers $k$ such that a subset $P_k$ as in the theorem exists.
We prove that $I = \{1, \ldots, n\}$ using Lemma~\ref{lem:ArithmeticLemma}.
To apply the lemma it suffices to show that $I$ fulfills the following properties:
 (i) $n \in I$,
(ii) If $k \in I$ then $n-k \in I$, and
(iii) If $k \in I$ then $k/2 \in I$.
Once we show that these properties hold, Lemma~\ref{lem:ArithmeticLemma} implies that any integer $k$ between 1 and $2n$ must be in $I$, hence $I = \{1, \ldots, n\}$.

Property (i) holds because the whole $S^1$ can be taken as a 2-arc set, containing $n$ points of each color, thus $n \in I$.

Property (ii) follows from the fact that the complement of any 2-arc set containing exactly $k$ points of each color is a 2-arc set containing exactly $n-k$ points of each color.
Thus if there is a 2-arc set guaranteeing that $k \in I$, its complement guarantees that $n-k \in I$.

Proving Property (iii) requires a more elaborate argument.
Let $\mathcal{A}_k$ be a $2$-arc set containing exactly $k$ points of each color (such a set must exist by hypothesis, since $k \in I$).
We assume $S^1$ is parameterized as $(\cos(t), \sin(t))$, for $t \in [0, 2\pi)$.
 Without loss of generality, we assume that $f(0)\not\in \mathcal{A}_k$ (if necessary we can change the parametrization of $S^1$ by moving the location of the point corresponding to $t=0$ to ensure this).

We lift all points of $P$ to $\mathbb{R}^3$ using the \emph{moment curve}, as explained next. Abusing slightly the notation, we identify each point $(\cos(t), \sin(t))$, $t \in [0, 2\pi)$ on $S^1$ with its corresponding parameter $t$. Then, for $t \in S^1$ we define $\gamma(t) = \{t, t^2, t^3\}$. Also, for any subset $\mathcal{C}$ of $S^1$, we define $\gamma(\mathcal{C}) = \{\gamma(p) | p \in \mathcal{C}\}$. Recall that we assumed that $f(0) \notin \mathcal{A}_k$, thus, $\gamma(\mathcal{A}_k)$ forms two disjoint arc-connected intervals in $\gamma(S^1)$.

Next we apply the ham-sandwich theorem to the points in $\gamma(\mathcal{A}_k)$ (disregarding other lifted points of $P$): we obtain a plane $H$ that cuts the three colored classes in $\gamma(\mathcal{A}_k)$ in half. That is, if $k$ is even, each one of the open halfspaces defined by $H$ contains exactly $k/2$ points of each color. If $k$ is odd, then we can force $H$ to pass trough one point of each color and leave $(k-1)/2$ points in each open halfspace (\cite{Matousek}, Cor. 3.1.3). We denote by $H^+$ and $H^-$ the open halfspaces above and below $H$, respectively. Note that each half space contains exactly $\lfloor k/2 \rfloor$ points of each color class, as desired.

Let $M_1 = H^+ \cap \gamma(\mathcal{A}_k)$ and $M_2 = H^- \cap \gamma(\mathcal{A}_k)$. Note that both $M_1$ and $M_2$ contain exactly $\lfloor k/2 \rfloor$ points of each color class, as desired. To finish the proof it is enough to show that either $\gamma^{-1}(M_1)$ or $\gamma^{-1}(M_2)$ is a $2$-arc set. Since $\gamma(\mathcal{A}_k)$ has two connected components (and is lifted to the moment curve), we conclude that any hyperplane (in particular $H$) can intersect $\gamma(\mathcal{A}_k)$ in at most $3$ points. Thus, the total number of components of $M_1 \cup M_2$ is at most 5. This is also true for the preimages $\gamma^{-1}(M_1) \cup \gamma^{-1}(M_2)$. Then, either  $\gamma^{-1}(M_1)$ or $\gamma^{-1}(M_2)$ must form a $2$-arc set containing exactly $\lfloor k/2 \rfloor$ points of each color.
Thus,  $\lfloor k/2 \rfloor\in I$ as desired.
\end{proof}

Our approach generalizes to $c$ colors: if $P$ contains $n$ points of each color on $S^1$, then for each $k\in \{1, \ldots, n\}$ there exists a $(c-1)$-arc set $P_k \subseteq S^1$ such that $P_k$ contains exactly $k$ points of each color. In our approach we lift the points to $\mathbb{R}^3$ because we have three colors, but in the general case we would lift to $\mathbb{R}^c$. We also note that the bound on the number of intervals is tight. Consider a set of points in $S^1$ in which the points of the first $c-1$ colors are contained in $c-1$ disjoint arcs (one for each color), and each two neighboring disjoint arcs are separated by $n/(c-1)$ points of the $c^{th}$ color. Then, if $k < n/(c-1)$, it is easy to see that we need at least $c-1$ arcs to get exactly $k$ points of each color.

We note that there exist several results in the literature that are similar to Theorem~\ref{thm:curveIntervals}. For example, in~\cite{Stromquist1985} they show that given $k$ probability measures on $S^1$, we can find a $c$-arc set whose measure is exactly $1/2$ in the $c$ measures. The methods used to prove their result are topological, while our approach is combinatorial. Our result can also be seen as a generalization of the well-known \emph{necklace theorem} for closed curves~\cite{Matousek}.

\section{$L$-lines in the plane lattice}\label{section:lattice}

We now consider a balanced partition problem for $3$-colored point sets in the integer plane lattice  $\mathbb{Z}^2$.
For simplicity, we will refer to $\mathbb{R}^2$ and $\mathbb{Z}^2$ as the plane and the lattice, respectively.
Recall that a set of points in the plane is said to be in \emph{general position} if no three of them are collinear.
When the points lie in the lattice, the expression is used differently: we say instead that a set of points $S$
in the lattice is in \emph{general position} when every vertical line and horizontal line contains at most one point from $S$.

An \emph{$L$-line with corner} $q\in\mathbb{R}^2$ is the union of two different rays with common apex $q$,
each of them being either vertical or horizontal.
An $L$-line partitions the plane into two regions (Figure~\ref{fig:L-line} shows a balanced $L$-line with apex $q$).
Since we look for balanced $L$-lines, we will only consider $L$-lines that do not contain any point of $S$.
Note that an $L$-line can always be slightly translated so that its apex is not in the lattice, thus its rays do not go through any lattice point.

\begin{figure}
 \begin{center}
\includegraphics*{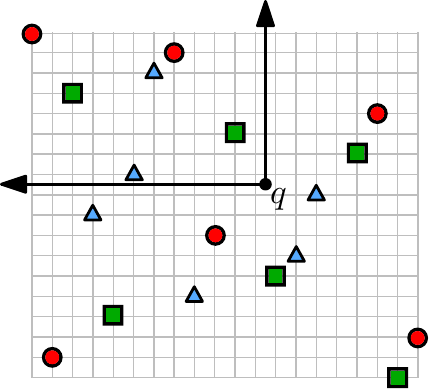}
 \end{center}
 \caption{A balanced set of $18$ points in the integer lattice with a nontrivial balanced $L$-line.}
\label{fig:L-line}
\end{figure}

$L$-lines in the lattice often play the role of regular lines in the Euclidean plane. For example, the classic ham-sandwich theorem (in its discrete version) bisects a 2-colored finite point set by a line in the plane.
Uno et al.~\cite{Kano2009} proved that, when points are located in the integer lattice, there exists a bisecting $L$-line as well.

Recently, the following result has been proved for bisecting lines in the plane.

\begin{theorem}[Bereg and Kano, \cite{Bereg2012}]\label{thm:balancedLine}
Let $S$ be a 3-colored balanced set of $3n$ points in general position in the plane.
If the convex hull of $S$ is monochromatic, then there exists a nontrivial balanced line.
\end{theorem}

As a means to further show the relationship between lines in the plane, and $L$-lines in the lattice, the objective of this section is to extend Theorem~\ref{thm:balancedLine} to the lattice. Replacing the term {\em line} for {\em $L$-line} in the above result does not suffice (see a counterexample in Figure~\ref{fig:diagonal}). In addition we must also use the orthogonal convex hull.

\begin{figure}
 \begin{center}
\includegraphics*{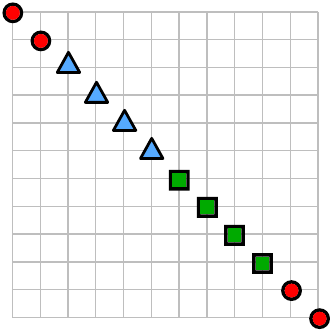}
 \end{center}
 \caption{A balanced set of 3-colored points in the plane
lattice. Any $L$-line containing points of all three colors will fully contain a color class, hence this problem instance does not admit a nontrivial balanced $L$-line.}
\label{fig:diagonal}
\end{figure}

\begin{theorem}\label{th:2lattice}
Let $S$ be a 3-colored balanced set of $3n$ points in general position in the integer lattice. If the orthogonal convex hull of $S$ is monochromatic, then there exists a nontrivial balanced $L$-line. 
\end{theorem}
\begin{proof}
Recall that a set $S\subset \mathbb{R}^2$ is said to be {\em orthoconvex} if the intersection of $S$ with every horizontal or vertical line is connected. The {\em orthogonal convex hull} of a set $S$ is the intersection of all connected orthogonally convex supersets of S.

Without loss of generality, we assume that the points on the orthogonal convex hull are red. We use a technique similar to that described in the proof of Theorem~\ref{th:bisectingdoubleWedge}; that is, we will create a sequence of orderings, and associate a polygonal curve to each such ordering. As in the previous case, we show that a curve can pass through the origin only at a vertex, which will correspond to the desired balanced $L$-line. Finally, we find two orderings that are reversed, which implies that some intermediate curve must pass through the origin. The difficulty in the adaptation of the proof lies in the construction of the orderings. This is, to the best of our knowledge, the first time that such an ordering is created for the lattice.

Given a point $p\in S$ we define the $0$-ordering of $p$ as follows. Consider the points above $p$ (including $p$) and sort them by decreasing $y$-coordinate, i.e., from top to bottom. Let $(p_1,\ldots, p_{j})$ be the  sorting obtained (notice that $p_j=p$). The remaining points (i.e., those strictly below $p$) are sorted by increasing $x$-coordinate, i.e., from left to right. Let $\sigma_{p,0}$ denote the sorting obtained.
 Similarly, we define the $\pi/2$, $\pi$ and $\frac{3\pi}{2}$-sided orderings $\sigma_{p,{\frac{\pi}{2}}}$, $\sigma_{p,\pi}$ and $\sigma_{p,{\frac{3\pi}{2}}}$, respectively. Each ordering can be obtained by computing $\sigma_{p,0}$ after having rotated clockwise the point set $S$ by $\pi/2$, $\pi$ or $3\pi/2$ radians, respectively. As an example, Figure~\ref{projection} shows $\sigma_{p,{\frac{3\pi}{2}}}$. Let $\mathcal{O}=\{\sigma_{p,i}\,|\, p\in S, i\in \{0,\frac{\pi}{2},\pi,\frac{3\pi}{2}\}\}$ be the collection of all such orderings of $S$.

\begin{figure}[h!]
\begin{center}
  \includegraphics{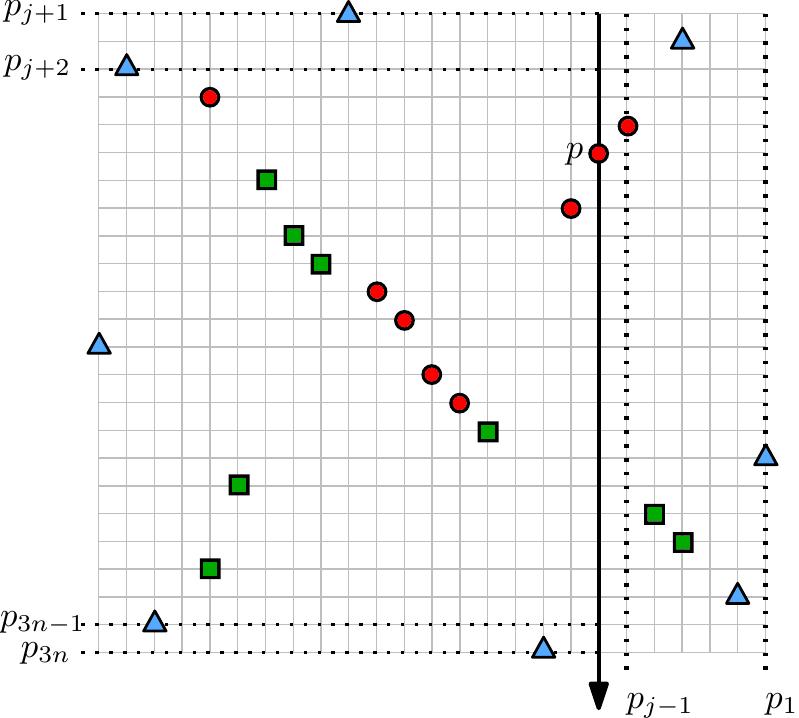}\\
  \caption{The $\frac{3\pi}{2}$-ordering of $S$ with respect to a given point $p$.}\label{projection}
\end{center}
\end{figure}

By construction, any prefix of an ordering of $\mathcal{O}$ corresponds to a set of points that can be separated with an $L$-line. Likewise, any $L$-line will appear as prefix of some sorting $\sigma\in \mathcal{O}$ (for example, one of the sortings associated with the apex of the $L$-line).

Given a sorting $\sigma=(p_1, \ldots, p_{3n})\in\mathcal{O}$, we associate it with a polygonal curve in the lattice as follows: for every $k \in \{1, \ldots, 3n-1\}$ let $b_k$ and $g_k$ be the number of blue and green points in $\{p_1\ldots, p_k\}$, respectively. Further, define the point $q_k := (3b_k -k, 3g_k-k)$.
Based on these points we define a polygonal curve $\phi(\sigma)=(q_1,\ldots, q_{3n-1}, -q_1, \ldots, -q_{3n-1})$. Similarly to the construction of Theorem~\ref{th:bisectingdoubleWedge}, the fact that $q_k=(0,0)$ for some $1 \leq k\leq 3n-1$ is equivalent to the fact that the corresponding $L$-line is balanced.
Therefore the goal of the proof is to show that there is always some ordering in $\mathcal{O}$ for which some $k$ has $q_k=(0,0)$.

We observe several important properties of $\phi(\sigma)$:

\begin{enumerate}
\item $\phi(\sigma)$ is centrally symmetric (w.r.t. the origin). This follows from the definition of $\phi$.

\item
\label{prop:int_segment}
The interior of any segment $q_iq_{i+1}$ of $\phi(\sigma)$ cannot contain the origin. The segment connecting two consecutive vertices in $\phi(\sigma)$ only depends on the color of the added element. Thus, there are only $3$ possible types of segments, see Figure~\ref{segments}. Since these segments do not pass through grid points, the origin cannot appear at the interior of a segment.

\item
\label{prop:winding}
For any $\sigma \in \mathcal{O}$, we have $q_1=(-1,-1)$, and $q_{3n-1}=(1,1)$. Notice that $q_1$ corresponds to an $L$-line having exactly one point on its upper/left side, while $q_{3n-1}$ corresponds to an $L$-line leaving exactly one point on its lower/right side. In particular, these points must belong to the orthogonal convex hull, and thus must be red. That is, $\phi(\sigma)$ is a continuous polygonal curve that starts at $(-1,-1)$, travels to $(1,1)$. The curve is symmetric and returns to $(-1,-1)$. Thus, as in the proof of Theorem~\ref{th:bisectingdoubleWedge} we conclude that either  $\phi(\sigma)$ passes through the origin or it has  nonzero winding.

\end{enumerate}

\begin{figure}[h!]
\begin{center}
  \includegraphics{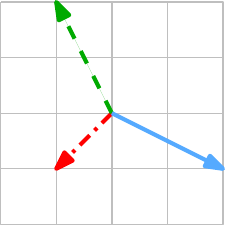}\\
  \caption{The three possible types of segments $q_{k-1}q_{k}$ depending on the color of $p_k$ (solid represents blue, dashed represents green and dotted dashed represents red).} \label{segments}
\end{center}
\end{figure}

Let $\sigma_x =(x_1, x_2, \ldots, x_{3n})$ be the points of $S$ sorted from left to right (analogously, $\sigma_y=(y_1, y_2, \ldots, y_{3n})$ for the points sorted from bottom to top). Observe that  $\sigma_x = \sigma_{x_{3n},\pi/2}$ and $\sigma_y= \sigma_{y_{3n},{\pi}}$ and their reverses are $\sigma^{-1}_x=\sigma_{x_{1}, \frac{3\pi}{2}}$ and $\sigma^{-1}_y=\sigma_{y_1, 0}$, respectively.

Analogous to the proof of Theorem~\ref{th:bisectingdoubleWedge}, our aim is to transform $\phi(\sigma_y)$ to its reverse ordering
through a series of small transformations of polygonal curves such that the winding number between the first and last curve is of  different sign.
If we imagine the succession of curves as a continuous transformation, a change in the sign of the winding number can only occur if at some point on the transformation the origin is contained in some curve. We will argue below that the only way in which the transformation passes through the origin is by having it as a vertex of one of the intermediate curves, which immediately leads to a balanced $L$-line.

Recall from Property~\ref{prop:int_segment} that the origin cannot be at the interior of an edge of any intermediate curve.
Thus, the only way the origin can be swept during the transformation is
(i) if it is a vertex of one of the curves,
or (ii) if it is contained in the space between two consecutive curves (``swept by the curves'').
In the remainder of the proof we show that the latter case cannot occur, that is, the origin is never swept by the local changes between two consecutive curves. This implies that the origin must be a vertex of some curve $\phi(\sigma)$ for some intermediate ordering $\sigma$, in turn implying that the associated $L$-line would be balanced.

The transformation we use is the following:

\begin{equation}
\begin{split}
\phi(\sigma_y)=\phi(\sigma_{y_{3n},\pi}) \rightarrow \phi(\sigma_{y_{3n-1},\pi}) \rightarrow \cdots \rightarrow \phi(\sigma_{y_{1},\pi}) \rightarrow \phi(\sigma^{-1}_x)\\
= \phi(\sigma_{x_{1},\frac{3\pi}{2}}) \rightarrow \cdots \rightarrow \phi(\sigma_{x_{3n},\frac{3\pi}{2}}) \rightarrow \phi(\sigma_{y_1,0})=\phi(\sigma_y^{-1})
\end{split}
\end{equation}

First we give a geometric interpretation of this sequence. Imagine sweeping the lattice with a horizontal line (from top to bottom). At any point of the sweep, we sort the points below the line from bottom to top, and the remaining points are sorted from right to left. By doing so, we would obtain the orderings $\sigma_y=\sigma_{y_{3n},\pi}, \ldots, \sigma_{y_{1},\pi}$, and (once we reach $y=-\infty$) $\sigma_x^{-1}$ (i.e., the reverse of $\sigma_x$). Afterwards, we rotate the line clockwise by $\pi/2$ radians, keeping all points of $S$ to the right of the line, and sweep from left to right. During this second sweep, we sort the points to the right of the line from right to left, and those to the left from top to bottom. By doing this we would obtain the orderings $\sigma^{-1}_x = \sigma_{x_{1},\frac{3\pi}{2}}, \ldots, \sigma_{x_{3n},\frac{3\pi}{2}}$. Once all points have been swept, this process will finish with the ordering $\sigma_{y_1,0}$, which is the reverse of $\sigma_y$.

Thus, to complete the proof it remains to show that the difference between any two consecutive orderings $\sigma$ and $\sigma'$ in the above sequence cannot contain the origin.

Observe that two consecutive orderings differ in at most the position of one point (the one that has just been swept by the line). Thus, there exist two indices $s$ and $t$ such that $s<t$ and $\sigma=(p_1, \ldots, p_s, \ldots, p_t, \ldots, p_{3n})$, and $\sigma'=(p_1, \ldots, p_s, p_t, p_{s+1}, \ldots p_{t-1}, p_{t+1}, \ldots, p_{3n})$ (that is, point $p_t$ moved immediately after $p_s$).

Abusing slightly the notation, we denote by $(q_1, \ldots, q_{3n-1},-q_1, \ldots, -q_{3n-1})$ the vertices of $\phi(\sigma)$ (respectively, $(q'_1, \ldots, q'_{3n-1},-q'_1, \ldots, -q'_{3n-1})$ denotes the vertices of $\phi(\sigma')$). Since only one point has changed its position in the ordering, we can explicitly obtain the differences between the two orderings. Given an index $i\leq 3n-1$, we define $c_i=(-1,-1)$ if $p_i$ is red, $c_i=(2,-1)$ if $p_i$ is blue, or $c_i=(-1,2)$ if $p_i$ is green. Then, we have the following relationship between the vertices of $\phi(\sigma)$ and $\phi(\sigma')$.

\begin{eqnarray}\label{eqntrans}
q'_i= \begin{cases}
q_i              &\mbox{if } i\in \{1, \ldots, s\} \cup \{t , \ldots, 3n-1\} \\
q_{i-1} +c_t &\mbox{if } i\in \{s+1, \ldots, t-1\}
\end{cases}
\end{eqnarray}

Observe that the ordering of the first $s+1$ points and the last $t-1$ points is equal in both permutations. In particular,  the points $q_1$ to $q_s$ and $q_{t}$ to $q_{3n-1}$ do not change between consecutive polygonal curves. For the intermediate indices, the transformation only depends on the color of $p_t$; it consists of a translation by the vector $c_t$.

We now show that the origin cannot be contained in the interior of the quadrilateral $Q_i$ of vertices $q_{i-1},q_i,q'_{i}$, and $q'_{i+1}$, for any $i\in\{s,\ldots, t\}$. Consider the case in which $i\in \{s+2, \ldots , t-1\}$ (that is, neither $i-1$ nor $i+1$ satisfy the first line of Equation~\ref{eqntrans}). Observe that the shape of the quadrilateral only depends on the color of $p_i$ and $p_t$. It is easy to see that when $p_i$ has the same color as $p_t$, $Q_i$ is degenerate and cannot contain the origin. Thus, there are six possible color combinations for $p_i$ and $p_t$ that yield three non-degenerate different quadrilaterals (see Figure~\ref{fig_square}).

\begin{figure}[h!]
\begin{center}
  \includegraphics{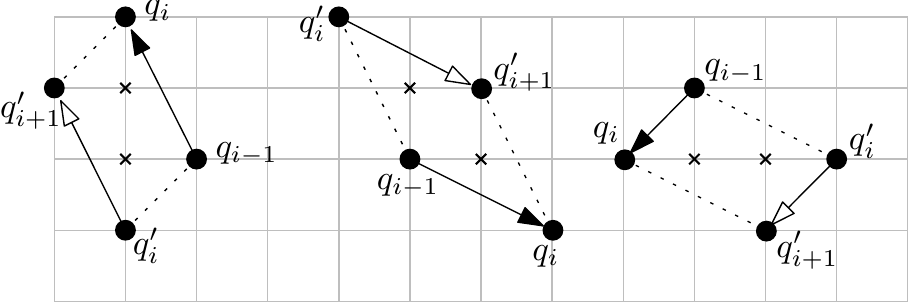}
  \caption{Three options for quadrilateral $Q_i$ depending on the colors of $p_i$ and $p_t$. From left to right $p_i$ is green, blue, and red whereas $p_t$ is red, green, and blue, respectively. The case in which $p_i$ and $p_t$ have reversed colors results in the same quadrilaterals. In all cases, the points in the lattice that are included in the quadrilateral are marked with a cross, thus those are the only possible locations for the origin to be contained in $Q_i$.} \label{fig_square}
\end{center}
\end{figure}

Assume that for some index $i$ we have that $p_t$ is red, $p_i$ is blue (as shown in Figure~\ref{fig_square}, left) and that the quadrilateral $Q_i$ contains the origin. Note that in this case, $q_i$ must be either $(0,1)$ or $(0,2)$. From the definition of the $x$-coordinate of $q_i$, we have $3b_i-k= 0$, and thus we conclude that $ k \equiv 0 \mod 3$. Consider now the $y$-coordinate of $q_i$; recall that this coordinate is equal to $3g_i - k$, which cannot be $1$ or $2$ whenever $k \equiv 0 \mod 3$. The proof for the other quadrilaterals is identical; in all cases, we show that either the $x$ or $y$ coordinate of a vertex of $Q_i$ is zero must simultaneously satisfy: $(i)$ it is congruent to zero modulo 3, and $(ii)$ it is either $1$ or $2$, resulting in a similar contradiction.

Thus, in order to complete the proof it remains to consider the cases in which $i=s+1$ or $i=t-1$. In the former case we have $q_{i-1}=q'_{i}$, and $q_{i}=q'_{i+1}$ in the latter. Whichever the case, quadrilateral $Q_i$ collapses to a triangle, and we have three non-trivial possible color combinations (see Figure~\ref{fig_triangle}). The same methodology shows that none of them can sweep through the origin. 

\begin{figure}[h!]
\begin{center}
  \includegraphics{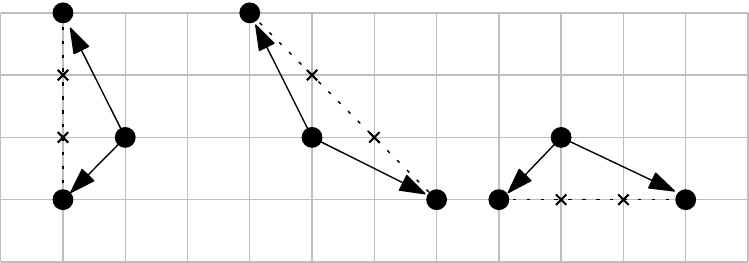}
  \caption{When $i=s+1$ or $i=t-1$ the corresponding quadrilateral $Q_i$ collapses to a triangle. As in the previous case, this results in three different non-degenerate cases, depending on the colors of $p_i$ and $p_t$.} \label{fig_triangle}
\end{center}
\end{figure}

That is, we have transformed the curve from $\phi(\sigma_{y_3n,\pi})$ to its reverse $\phi(\sigma_{y_1, 0})$ in a way that the origin cannot be contained between two successive curves. However, since these curves have winding number of different sign, at some point in our transformation one of the curves must have passed through the origin. The previous arguments show that this cannot have happened at the interior of an edge or at the interior of a quadrilateral between edges of two consecutive curves. Thus it must have happened at a vertex of $\phi(\sigma)$, for some $\sigma\in\mathcal{O}$. In particular, the corresponding $L$-line must be balanced.
\end{proof}

\section{Concluding remarks}\label{section:conclusion}
In this paper we have studied several problems about balanced partitions of $3$-colored sets of points and lines in the plane. As a final remark we observe that our results on double wedges can be viewed as partial answers to the following interesting open problem: Find all $k$ such that, for any $3$-colored balanced set of $3n$ points in general position in the plane, there exists a double wedge containing exactly $k$ points of each color. We have given here an affirmative answer for $k=1,n/2$ and $n-1$ (Theorems \ref{thm:doubleWedge111} and \ref{th:bisectingdoubleWedge}). Theorem \ref{thm:curveIntervals} gives the affirmative answer for all values of $k$ under the constraint that points are in convex position.

\bigskip

\noindent{\bf Acknowledgments} Most of this research took place during the \emph{5th International Workshop on
Combinatorial Geometry} held at UPC-Barcelona from May 30 to June 22, 2012. We thank In\^{e}s Matos, Pablo P\'{e}rez-Lantero, Vera Sacrist\'{a}n and all other
participants for useful discussions.

F.H., M.K., C.S., and R.S. were partially supported by projects MINECO MTM2012-30951, Gen. Cat. DGR2009SGR1040 and DGR2014SGR46,
and ESF EUROCORES programme EuroGIGA, CRP ComPoSe: MICINN Project EUI-EURC-2011-4306.
F.H. and C.S. were partially supported by project MICINN MTM2009-07242.
M.K. also acknowledges the support of the Secretary for Universities and Research of the Ministry of Economy and Knowledge of the Government of Catalonia and the European Union.
R.~I.~S. was funded by Portuguese funds through CIDMA and FCT, within project PEst-OE/MAT/UI4106/2014, and by FCT grant SFRH/BPD/88455/2012.

\end{document}